\documentclass[12pt]{article}
\usepackage[cm]{fullpage}
\usepackage{graphicx}
\usepackage{caption}
\usepackage{subcaption}
\usepackage{mwe} 
\usepackage{titlesec}
\usepackage{lipsum}
\usepackage{subcaption}
\usepackage[utf8]{inputenc}
\usepackage{tikz}
\usepackage{pgfplots}
\usetikzlibrary{arrows}
\usepackage{latexsym}
\usepackage{amsthm}
\usepackage{amsthm}
\usepackage{comment}
\usepackage{hyperref}
\usepackage{mathrsfs}  
\usepackage{color}
\usepackage{filecontents}
\graphicspath{{pictures/}}
\usepackage{tikz}
\usetikzlibrary{shapes,arrows}
\usepackage{verbatim}
\usepackage{mathtools}
\usepackage{float}
\usepackage{fdsymbol}
\usetikzlibrary{arrows.meta,decorations.markings}
\usepackage[linesnumbered,algoruled,boxed,lined]{algorithm2e}

\usepackage[
backend=biber,
style=numeric,
maxbibnames=99,
sorting=ynt
]{biblatex}

 \newcommand{\PROB}{$Dmax_{\ast}(S,m)$}
  \newcommand{\PROBS}{$Smax_{\ast}(S,m)$}
  
    \newcommand{\PROBSS}{$Smax_{\Box}(S,m)$}

\newcommand{\MWPIHP}{\textit{MWPIHP}}
 
  \newcommand{\NAMEH}{$Dmax_{\Box}(S,m)$}
  \newcommand{ \NAMED}{$Dmax_{\medcircle}(S,m)$}
    \newcommand{\NAMEHr}{$RSmax_{\Box}(S,m)$}
   \newcommand{\NAMEHS}{$HSmax_{\Box}(S,m)$}

   \newcommand{ \NAMEDS}{$Smax_{\medcircle}(S,m)$}
   \newcommand{ \NAMEDSH}{$HSmax_{\medcircle}(S,m)$}
 \newcommand{\minCell}{$minC$}
 \newcommand{\maxCell}{$maxC$}
  \newcommand{\cell}{\textit{CELL}}

\newtheorem{theorem}{Theorem}
\begin{document}

\title{Construction and Maintenance of Swarm Drones\thanks{This research has been supported by the grant from Pazy Foundation}}
\author{Kiril  Danilchenko \\\href{mailto:kirild@post.bgu.ac.il}{kirild@post.bgu.ac.il} \and
Michael Segal \\ \href{mailto:kirild@post.bgu.ac.il}{segal@bgu.ac.il} }
\date{%
    School of Electrical and Computer Engineering \\ 
Ben-Gurion University of the Negev, Beer-Sheva, Israel%
}

%
%
%
%

\maketitle              
\begin{abstract}
 In this paper we study the dynamic version of the covering problem motivated by the coverage of drones' swarm: Let $S$ be a set of $n$ non-negative weighted points in the plane representing users. Also consider  a set  $P$ of $m$ disks that correspond to the covering radius of each drone. We want to place (and maintain) set $P$ such that the sum of the weights of the points in $S$ covered by disks from $P$ is maximized. We present a data structure that maintains a small constant factor approximate solution efficiently, under insertions and deletions of points to/from $S$ where each update operation can be performed $O(\log n)$ time.

\end{abstract}

\section{Introduction}
Unmanned aerial vehicles (UAVs or drones) have been the subject of concerted research over the past few years. UAVs have many potential applications in wireless communication systems \cite{zeng2016wireless}. In particular, UAV-mounted mobile base stations (MBSs).
UAVs-mounted MBSs can be used to provide wireless coverage in several scenarios: emergency cases, battlefields, and more. The performance of the UAV-based network depends on the deployment of UAVs. Generally, the drone deployment problem has been studied to find out the optimal 3D positions for drones to serve ground users on a 2D plane.
The researchers have focused on deployment strategies based on minimizing the number of UAVs required to provide wireless coverage to all ground users.  

Opposite to the most existing work where the number of UAVs is assumed to be at least as a number of ground users \cite{lyu2016placement,mozaffari2016efficient,pugliese2016modelling}, here we consider a more realistic scenario. We focus on the situation where the number of UAVs is given, and this number is significantly less than the number of ground users. This assumption is reasonable in emergency cases or battlefields where the number of ground users (for example, soldiers or firefighters) is more than the number of UAVs. 
As a result, we face a problem that not all users will necessarily be covered. Additionally,  we assume that the ground users have a rank. For example, on the battlefield, it is more critical to provide wireless coverage to the commanding officer than to the regular soldier. The number, rank, and location of ground users may be changed at a time to time.
Thus, we are facing the most crucial question "who should be covered and who should not, at any point of time?" For a given particular snapshot, this problem is known to be NP-hard \cite{de2009covering}. In this paper, we assume that all UAVs fly at the same, fixed altitude. See Figure \ref{fig:fig}.
To deal with our problem, we presented a dedicated data structure and a constant factor approximation algorithm under addition, deletion, or rank change of ground user. 
\begin{figure}[H]
\centering
   \includegraphics[width=.7\linewidth]{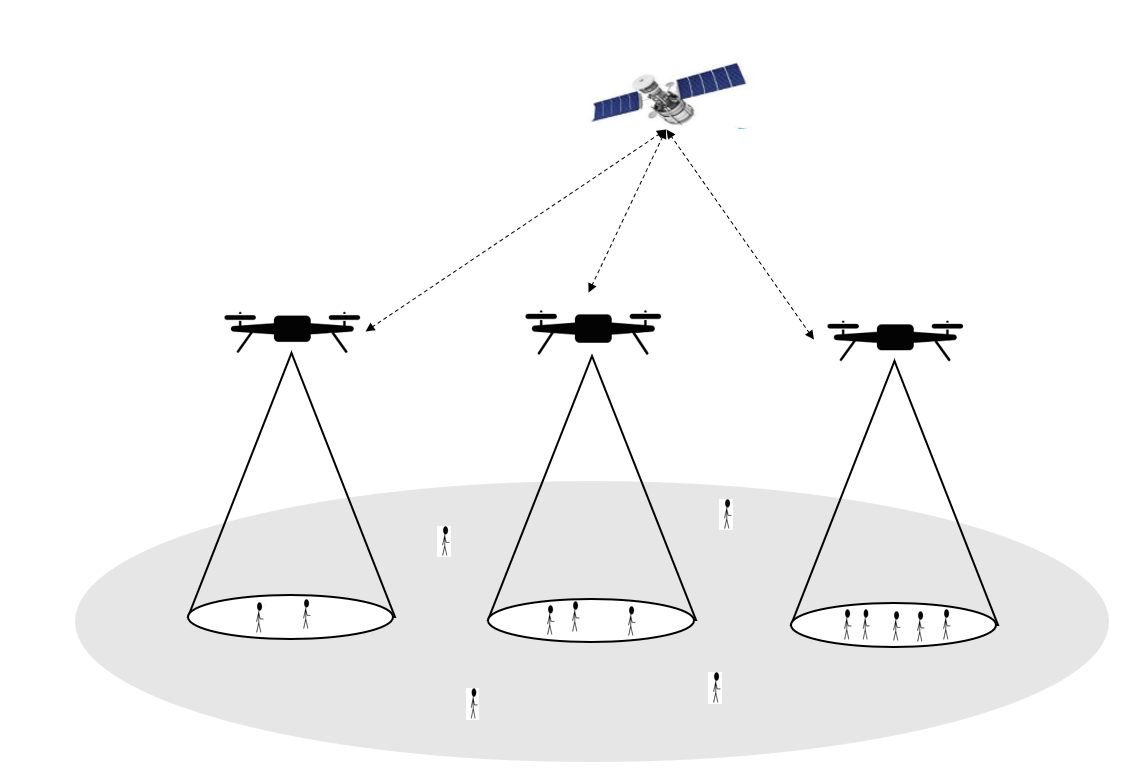}
  
\caption{Coverage of ground users by UAVs.}
\label{fig:fig}
\end{figure}
The rest of the paper is organized as follows. Section 2 
reviews the related work. In Section 3, we formally describe the problem. Section 4 presents the approach for calculating the upper bound value of our problem.
The proposed solutions for the scenario where the ground users are static are shown at Section 5. Section 6 presents solutions for dynamic version. Finally, Section 7 concludes the paper and discusses the future work.
\vspace{-0.4cm}
\section{Related Work}
\vspace{-0.2cm}
The authors of \cite{lyu2016placement} aimed to provide wireless coverage for a group of ground terminals, ensuring that each
ground terminal is within the communication range of at
least one UAV with minimal number of UAVs. They proposed a polynomial-time algorithm
with successive UAVs placement, where the UAVs are placed
sequentially starting from the area perimeter of the uncovered
ground terminals along a spiral path towards the center, until
all ground terminals are covered. The authors in \cite{mozaffari2016efficient} studied the optimal deployment of UAVs, such that  the total coverage area is
maximized. The UAVs were assume to be equipped with directional antennas.

The authors of \cite{pugliese2016modelling} and  \cite{zorbas2016optimal} required to provide drones deployments in which a given set of objects should be
covered by the sensing range of at least one drone. Each object should be monitored for a certain amount of time. The authors consider minimizing the number of drones to cover the target and the total energy consumption of the drones.

The paper \cite{zhao2018deployment} investigated a challenging problem of how to deploy multiple UAVs
for on-demand coverage while at the same time maintaining the
connectivity among UAVs. The paper \cite{srinivas2009construction} utilizes autonomous mobile base stations (MBSs) to maintain network connectivity while minimizing the number of MBNs that are deployed. They \cite{srinivas2009construction} formulate the problem of reducing the number of MBSs and refer to it as the Connected Disk Cover (CDC) problem. They show that it can be decomposed into the Geometric Disk Cover (GDC) problem and the Steiner Tree Problem with Minimum Number of Steiner Points (STP-MSP). 
The authors proved that if these subproblems can be solved separately by $\gamma$-approximation and $\delta$-approximation algorithms, respectively, the approximation ratio of the joint solution is $\gamma+\delta$.

The authors of \cite{de2009covering} were the first who presented the problem of covering the maximum number of points in the point set $S$ with $m$ unit disks. They call this problem $max(S,m)$. They gave the first $(1 -\varepsilon)$-approximation algorithm for the problem $max(S,m)$ with time complexity $O(n\varepsilon^{-4m+4})\log^{2m-1}(\frac{1}{\varepsilon })$. The problem to place $m$  rectangles such
that the sum of the weights of the points in $P$ covered by these rectangles is maximized is considered in \cite{li2015linear}. For any fixed $\varepsilon>0$, the authors present efficient approximation
schemes that can find a $(1 -\varepsilon)$-approximation to the optimal solution in $O(\frac{n}{\varepsilon }\log (\frac{1}{\varepsilon })+m(\frac{1}{\varepsilon })^{O(\min (\sqrt{m},\frac{1}{\varepsilon }))})$ run time. In \cite{jin2018near} the authors presented a PTAS for a more general case different covering shapes (disks, polygons with $O(1)$ edges), running in $O(n\frac{1}{\varepsilon }^{O(1)}+\frac{m}{\varepsilon }\log {m}+m(\frac{1}{\varepsilon })^{O(\min (m,\frac{1}{\varepsilon }))})$ time.
Regarding 1-dimensional case, the authors of \cite{chrobak2015scheduling} have shown how to find $\delta$ points on the line that hit a maximum number of intervals, for a given family of $n$ intervals, in $O(\delta n^2)$ time. The paper \cite{damaschke2017refined}  proved that the problem from \cite{chrobak2015scheduling} is equivalent to finding $\delta$ cliques in a $n$ interval graph that cover a maximum number of distinct vertices. In \cite{damaschke2017refined} the authors  show $O(\delta|E|)$ running time   for connected interval graphs. 
\vspace{-0.5cm}
\section{Problem Formulation}
\vspace{-0.2cm}
 
We consider a set $S$  of $n$ points  distributed in the plane, where each point $s_i\in S,i=1,\dots,n,$ has a positive weight $w(s_i)$. 
Also consider  a set  $P$ of $m$ disks (squares) of radius $R_{COV}$.
We define the following problems: \\ 
 \textbf{ Static-max(S,m)}(\PROBS): Given a set $S$, place shapes from the set $P$ such that the total weight of points from the set $S$ covered by the disks is maximized.\\
   \textbf{ Dynamic-max(S,m)}(\PROB): Given a set $S$, maintain shapes from the set $P$ such that the sum of the weights of the points in $S$ covered by disks from $P$ is maximized  under insertions and deletions of points to/from $S$. \\
  Note that "$\ast$" symbolizes which covering shape is in use ($\Box$ stands for square and $\medcircle$ stands for disk).
  \vspace{-0.5cm}
  \section{One-dimensional Case}
  \vspace{-0.2cm}
In this section, we present the approach for calculating the upper bound value of \PROBS\, and the required runtime  to calculate it.  To evaluate the maximum possible weight of covered points from $S$, we make dimension reduction and deal with the one-dimensional instance of the \PROBS. The one-dimensional instance of the \PROBS\, is defined as follows. 
Given a set $S$ of $n$ weighted points on a line, and a set $P$ of $m$ intervals having length $2R_{COV}$, we would like to locate intervals from $P$ to cover a subset of $S$ having maximum weight. Note that this problem is equivalent to the following piercing problem:
Given a set $S'$ of $n$ weighted intervals $i_{j}=\left[l_{j},r_{j}\right],j=1,\ldots,n,$ of length $2R_{COV}$, we want to identify a set $P'$ of $m$ piercing points that pierce a subset of $S'$ having maximum weight. The weight of $i_j$ denoted by $w(i_j)$ and is equal to $w(s_j),s_j \in S,j=1,\dots,n.$




We call  this problem $m$-\textit{Maximum Weighted Partial Interval Hitting Problem} (\MWPIHP). 
We define a neighbour of interval $i_j$ to be interval $s_k \in S'$ such that $s_k \cup s_j \neq \emptyset$. The neighbourhood of an interval $i_j$ in  $S'$ is the subset of $S'$ induced by all neighbour intervals of $i_j$, including $i_j$ itself.
We solve  \MWPIHP\, by using a dynamic programming approach. Let   $F\left(j,k\right)$  represent the maximum weight of piercing a set $\{i_1,\dots,i_j\}$ of intervals by
$k$ piercing points. $F(j,k)$ can be computed by considering two cases: whether the $k^{th}$ piercing point is used to pierce $i_j$ or not.

Thus, we have the following recurrence for a set of sorted (by left endpoint) intervals $\{i_1,\dots,i_n\}$.
\vspace{-0.2cm}
\begin{equation}
  F\left(j,k\right) =
    \begin{cases}
      0 & \text{if $k=0$}\\
      \max\left\{ F\left(j-1,k\right),F\left(j-N,k-1\right)+W\left(l_j\right)\right\}  & \text{if $k\geq 1$,}\\
      & 1 \leq j\leq n,\\
      &1 \leq k\leq m
     
    \end{cases}       
\end{equation}

Using the fact that all intervals have equal length we can assume that we always pierce $i_{j}$ on its left endpoint. The $W(l_j)$ value is the weight of all intervals pierced by point $l_j$, and $N$ value is the number of intervals pierced by point $l_j$.
Computation of $W(i_j)$ and $N$ can be done in $O\left(\log n\right)$ time by following approach.

We build a balanced binary search tree on the left endpoints of the intervals $\{l_k\}_{k=1}^n$, where each node keeps the total weight of its left subtree and its right subtree, as well the number of nodes in its subtrees. The weight of leaf is equal to the weight of intervals it belongs to, and the number of nodes on its left and right subtree is equal to zero.
In this tree, we look for nodes that belong to the range between the left endpoint of the interval $i_j$ and the left endpoint of a leftmost neighbor of $i_j$. In order to find these nodes we look for the leaf with value $v_{l_j}$ and a leaf with the smallest value that is equal or larger than $v_{l_j-2R_{COV}}$. Next we find the lowest common ancestor of identified leaves $v_{l_j}$ and $v_{l_j-2R_{COV}}$. Denote this node as $lca$.
Define the set of nodes that are located in the path from $v_{l_j}$ to $lca$ as $R$. We will sum up the values of the right subtrees of the trees rooted by the nodes from $R$, denote this sum as $sum_r$. Symmetrically we perform the similar actions from $v_{l_j-2R_{COV}}$. Denote this sum as $sum_r$.
The value of  $W(i_j)$ is equal to $W(i_j) =sum_r+sum_l$.
We find $N$ similarly to $W(i_j)$ in $O(\log n)$ runtime.
This gives us an algorithm with  the running time $O(mn\log n)$, since we have $O(mn)$ values to compute, where each computation takes time $O(\log n)$.

  \begin{theorem}
  The upper bound of \PROBS\, can be evaluated in $O(mn\log n)$ runtime.
  \end{theorem}
  \begin{proof}
 The upper bound of \PROB, i.e., the maximum weight of points which can be covered by $m$ squares can be calculated by projection  of   the set $S$ on $x$ and $y$ axis, following the solution of \MWPIHP\, for each axis separately by choosing the maximum weight between the two results.


  \end{proof}
\vspace{-0.5cm}
\section{Two-dimensional Case}
\vspace{-0.2cm}
First, we focus on the case of squares. The \PROBSS\ is NP-hard when $m$ is part of the input, and the shape is square \cite{megiddo1984complexity,jin2018near}. We slightly change the definition of \PROBSS. The relaxed \PROBSS\, contains an additional constraint required that the distance between each pair of points in $C$ is at least  $2 R_{COV}$, where the set $C$ contains the centers of squares from set $P$. This is in order to avoid any overlapping in the coverage (see \cite{mozaffari2016efficient}) and as a result to prevent interference between neighboring UAVs. We denote this instance of problem \PROBSS\,   as \NAMEHr.\\
  \begin{theorem}
 The \NAMEHr\ is NP-hard.
\end{theorem}
\begin{proof}
The reduction is from the \textit{(p,k)-Rectangle Covering problem} \cite{ahn2011covering}, which is NP-hard problem for any fixed $k$ when $p$ is part of the input. \textit{(p,k)-Rectangle Covering problem} finds $p$ axis-aligned, pairwise-disjoint boxes that together contain at least $n-k$ points. We focus on the decision version of the  problem \textit{(p,0)-Rectangle Covering problem}: Given $n$ points in the plane and an integer $p > 0$, decide whether or not there exist $p$ axis-aligned unit squares  that together cover all points.
The reduction is as follows. The parameter $R_{COV}$ is equal to $0.5$ and the weight of each point in the set $S$ is equal to one. The solution of the \NAMEHr\, will be the number of covered points. Therefore, the number of covered points equals to $n$ iff it is possible to cover all the points with $p$ unit non-overlapping squares.

\end{proof}
\subsection{Approximation solution for \PROBSS}
We want to mention that our primary goal is to find the dynamic solution (under insertions/deletions of weighted points from $S$). Unfortunately, it is impossible to make the dynamic solutions based on the techniques presented in \cite{li2015linear,jin2018near} for achieving PTAS. Therefore,
we first present the solution for the static case and then show how to make it work in the dynamic setting.
  Denote by $G_{r}$ the square grid with cell size $r$ such that the vertical and the horizontal lines are defined as follows:
\begin{align}\label{eq:Grid}
G_{r}=\left\{ \notag (x,y)\in \mathbb{R}^2 \mid x=k\cdot r,k \in \mathbb{Z}\right \} \bigcup \left\{ \notag (x,y)\in \mathbb{R}^2 \mid y=k\cdot r,k \in \mathbb{Z}\right \}
\end{align}
Given a point $s_i \in S$ we call the integer pair  $\left( \left\lfloor\dfrac{x}{r}\right\rfloor ,\left\lfloor\dfrac{y}{r}\right\rfloor\right)$ as \textit{cell index}  (the cell in which $s_i$ is located). Each nonempty cell in $G_r$ will be identified by \textit{index}. 
We calculate the \textit{cell index} $(a_i,b_i)$ for each point $s_i \in S$
and use $(a_i,b_i)$ to find \textit{index} $\pi(a_i,b_i)$. The \textit{index} $\pi(a_i,b_i)$ is a result of Cantor pair function \cite{szudzik2006elegant}, where the \textit{cell index} $(a_i,b_i)$ is input of this function. Cantor pairing function is a  recursive pairing function $\pi:\mathbb{N} \times \mathbb{N} \rightarrow \mathbb{N}$ defined by
 \begin{equation}\label{eq:cantor}
     \pi(x,y)=\frac{(x+y+1)(x+y)}{2}+y
 \end{equation}

Next, we present our heuristic solution of the \PROBSS\ and  \NAMEHr. Denote this solution as \NAMEHS.
We divide the area that contains the set $S$ into grid with cell size $r=2R_{COV}$. We  place the set of $m$ squares of $P$ in the cells having maximal weight, such that the squares fit the cells of $G_r$.

\begin{algorithm}[H]

\RestyleAlgo{ruled}
\KwIn{$G_r$, $m$}
\KwOut{Set $P$ covering $m$ cells having largest weight}
Sort the weights of $G_r$  cells in non decreasing order.\\
Place the set $P$ in the $m$ cells having largest weight.
 \caption{ \NAMEHS } 
 \label{Alg:Intersection} 
\end{algorithm}
\begin{theorem}\label{THM:Appr}
The \NAMEHS\ algorithm provides 1/4-approximate solution to \PROBSS.
\end{theorem}
    \begin{proof}
 Define the optimal solution of \PROBSS\ by $OPT$  and the solution that is achieved by  \NAMEHS\ as $SOL$. Additionally, denote the total weight of points covered by $OPT$ and $SOL$ as $W(OPT)$ and $W(SOL)$, respectively.
Each of the squares in $OPT$ may overlap with at most 4 cells from $G_r$. In this case grid $G_r$ divides each square into four rectangles. In Figure \ref{Fig:apro} grid $G_r$ divides one of the squares from $OPT$, $R_{o}$  into four rectangles: $o_1$, $o_2$, $o_3$  and $o_4$.
As already mentioned, each square from $OPT$ is divided by at most four rectangles. Therefore the number of cells in $G_r$ overlapping with the squares from $OPT$ is $4m$. Define the set of sorted cells that overlap with $OPT$, as $Q$, and the weight of this set as $W(Q)$. Additionally, define the set of $m$ cells having maximal weight from set $Q$ as $Q_m$ and the weight of this set as $W(Q_m)$. Obviously, the weight of $W(Q_m)\geq \frac{W(Q)}{4}$. It is easy to see that the equality holds only in the case of uniform distribution of weight among the cells from $Q$; in any other case, the inequality holds. 
On the other side, the weight of $SOL$ is greater than or equal to the weight of any other set of  cells from $G_r$ with the same amount of cells. Therefore $ W(SOL)\geq W(Q_m)\geq \frac{W(Q)}{4}$.

 \begin{figure}[h] 
\centering
\begin{tikzpicture}
\draw[step=1cm,gray,very thin] (-0.99,-0.99) grid (4.99,4.99);

 \draw (0.75,1.5) -- (0.75,2.50) -- (1.75,2.5) -- (1.75,1.5) -- (0.75,1.5); 
 \foreach \x in {1,...,550}
    {
      \pgfmathrandominteger{\a}{-100}{490}
      \pgfmathrandominteger{\b}{-100}{490}
      \fill [fill=red](\a*0.01,\b*0.01) circle (0.02);
    };

    \foreach \x in {1,...,30}
    {
      \pgfmathrandominteger{\a}{75}{155}
      \pgfmathrandominteger{\b}{165}{250}
      \fill [fill=red](\a*0.01,\b*0.01) circle (0.02);
    };
    
    \node(start)[below] at (-0.5,0.75) {$o_1$};
     \node(start2)[below] at (-0.5,3.5) {$o_2$};
     \node(start3)[below] at (2.5,0.75) {$o_3$};
     \node(start4)[below] at (2.5,3.5) {$o_4$};
     \node(start5)[below] at (3.3,2.32) {$R_{o}$};
     (-0.20,0.70) 
     \draw [black, ->          ](start) -- (0.85,1.65);
      \draw [black, ->          ](start2) -- (0.85,2.25);
      
      \draw [black, ->          ](start3) -- (1.65,1.65);
      \draw [black, ->          ](start4) -- (1.65,2.25);
     \draw [black, ->          ](start5) -- (1.78,2);
     
      \draw [black, <->          ](0,-0.5) -- (1.0,-0.5);
      \node(radius)[below] at (0.5,-0.5) {$2 R_{COV}$};
\end{tikzpicture}
\caption{Overlapping between the square from $OPT$ and the grid $G_r$}
\label{Fig:apro} 
\end{figure}
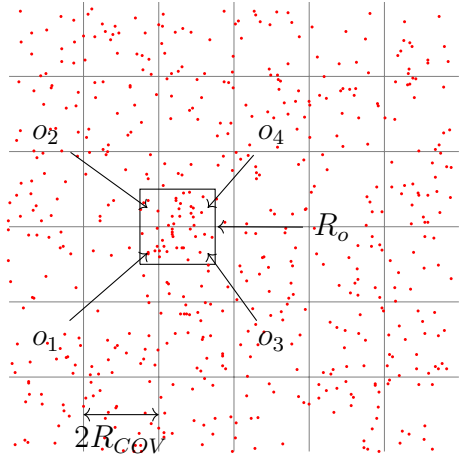

       \end{proof}
    In order to find our solution, we build a  balanced binary search tree $T$ where each node has two fields: key and value. The key field is an \textit{index} of non-empty  cell in grid $G_r$ 
 and value is a total weight of points covered by this cell. For each point $s_i \in S$ we calculate \textit{cell index} $(a_i,b_i)$ and use  $(a_i,b_i)$ to find  $\pi(a_i,b_i)$. For each point $s_i$ we insert  node with key $\pi(a_i,b_i)$ in the $T$, and if it is necessary update the value of this node. 
 
 For each cell in $G_r$ we calculate the total weight of points belonging to this cell; this calculation takes $O\left( n \log n \right)$ runtime. Note that the number of non-empty cells is at most $n$. The runtime required to built a $T$  is  $O\left(n \log n\right)$ and to find $m$ most weighted cells takes $O\left(m \log n\right)$ runtime. 
     \begin{theorem}\label{THM:ApprDynami}
The optimal solution of \PROBSS\, can be achieved by $HSmax_{\Box}(S,4m)$.
\end{theorem}
\begin{proof}
 \vspace{-0.2cm}
We use similar notations as in Theorem \ref{THM:Appr}. The size of $Q$ is at most $4m$ and the size of $SOL$ is $4m$, too. The weight of $SOL$ is greater than or equal to the weight of any other set of  cells from $G_r$ with the same amount of cells. Obviously, $W(SOL) \geq W(Q)$. 

\end{proof}

\vspace{-0.5cm}
\subsection{Approximate solution for \NAMEDS }
\vspace{-0.2cm}
Now, we focus on the case of disks. We use similar notations as in previous section. We change the size of cell in grid $G_r$ to be  $r=\sqrt{2}R_{COV}$. All the rest remain the same. Denote this solution as \NAMEDSH.

\vspace{-0.2cm}
\begin{theorem}\label{THM:Apprd}
The \NAMEDSH\ algorithm provides $1/7$-approximate solution to $Smax_{\medcircle}(S,m)$.
\end{theorem}
\begin{proof}
\vspace{-0.1cm}
The proof is  similar to the proof of Theorem \ref{THM:Appr} with a slight difference. Each of the disks in $OPT$ may overlap with at most 7 cells from $G_r$. See Figure \ref{Fig:apro}. Thus, set $Q$ may overlap with $7m$ cells. Using similar arguments as in Theorem \ref{THM:Appr} we can see that the approximation ration is $\frac{1}{7}$.
\vspace{-0.3cm}
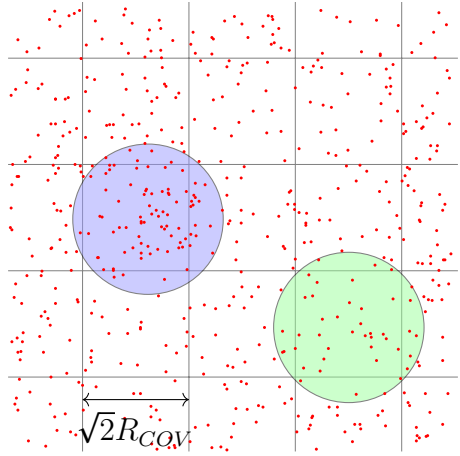
\begin{figure}[h] 
\centering
\begin{tikzpicture}
\draw[step=1.41421cm,gray,very thin] (-0.99,-0.99) grid (4.99,4.99);

\draw [fill=blue!40,opacity=0.5](0.87,2.1) circle (1);
\draw [fill=green!40,opacity=0.5] (3.54,0.66) circle (1);
 \foreach \x in {1,...,550}
    {
      \pgfmathrandominteger{\a}{-100}{490}
      \pgfmathrandominteger{\b}{-100}{490}
      \fill [fill=red](\a*0.01,\b*0.01) circle (0.02);
    };

    \foreach \x in {1,...,30}
    {
      \pgfmathrandominteger{\a}{75}{155}
      \pgfmathrandominteger{\b}{165}{250}
      \fill [fill=red](\a*0.01,\b*0.01) circle (0.02);
    };
    
      
     
      \draw [black, <->          ](0,-0.3) -- (1.41421,-0.3);
      \node(radius)[below] at (0.7,-0.3) {$\sqrt{2}R_{COV}$};
\end{tikzpicture}
\caption{Location of disk from optimal solution and $HSmax_{\medcircle}(S,m)$, left  and right disks, respectively}
\label{Fig:apro} 
\end{figure}
\end{proof}

\section{Maintenance of Dynamic Covering Set}
\vspace{-5pt}
 In this section, we deal with a dynamic set of points $S$. From time to time, a point is added to $S$, deleted from $S$, or its weight may be updated. In the dynamic version of the \PROB\ problem we want to maintain dynamically a set $P$ of disks (squares) to maximize the total weight of covered points from $S$.
 
First, we focus on the case of squares. To deal with the dynamic set $S$, we offer the following approach. After each change in $S$, we check if the square $p_i \in P$ that covers points having minimal total weight over all squares from $P$ after the set update needs to  change its location in order to cover points of larger total weight. At the start of the algorithm, we place the set $P$ by  \NAMEHS\ from the previous section. For simplicity, we also assume that all the  points from $S$ are in general position.



 
   
We construct a data structure $D$ of size $O\left(n \right)$ that will allow us to update the set $S$  and to maintain the location of $P$ dynamically. The data structure $D$  supports the following operations:
\begin{itemize}
    \item  $Insert\left(s_{new}\right)$ - Insert a new point $s_{new}$ into $S$.
    \item  $Delete\left(s_i\right)$ - Delete an existing point $s_i$ from $S$.
    \item  $Update\left(s_i,\hat{w}\right)$ - Given $s_i \in S$ update the weight of  $s_i$ to be $\hat{w}$.
\end{itemize}
 Among the non-empty cells of grid $G_r$ and covered cells  by set $P$, we select the cell \minCell\ having the minimum weight. Also, among the non-empty cells of grid $G_r$ and not covered cells by set $P$, we select the cell \maxCell\ having the maximum weight. We define the weight of the cell in $G_r$ with \textit{index} $\pi(a,b)$ as $ \rho_{\pi(a,b)}$.

The proposed data structure $D$ is a combination of two data structures. The first data structure keeps the set $S$, and the second data structure contains the set $C$ of square centers and weights of points covered by each $p_i \in P$. We denote these data structures by $D_1$ and $D_2$, respectively.

The first  data structure $D_1$ is a balanced binary search tree built on the $x$-coordinate of the points in $S$.

 
The second data structure $D_2$ is combination of three balanced binary search trees $T_1$, $T_2$ and $T_3$. The balanced binary search tree $T_1$ is a tree where each node corresponds to the weight of the cell from $G_r$ covered by the square from set $P$. 
Therefore if cell with \textit{index} $\pi(a,b)$ covered by the square from set $P$, then node with key $ \rho_{\pi(a,b)}$ belongs to $T_1$.
The balanced binary search tree $T_2$ is a tree where each node corresponds to the weight of nonempty cell in $G_r$, which is not covered by any square from set $P$. Therefore if nonempty cell with \textit{index} $\pi(a,b)$ not covered by any square from set $P$, then node with key $ \rho_{\pi(a,b)}$ belongs to $T_2$.
The balanced binary search tree $T_3$ is a tree where each node has key and value. The key is an \textit{index} $\pi(a,b)$ of the non-empty cell in $G_{r}$, and value is a  $c_i \in C$ if the cell is covered by square $p_i \in P$ or $NULL$  if the cell is not covered by square from set $P$.
Additionally, each node in $T_3$ has a pointer to node in $T_1$ or $T_2$, accordingly if the cell is covered or not. We note that there is no need to store any information for empty cells. See Figure \ref{Fig:Data structure}.
    \vspace{-0.6cm}
 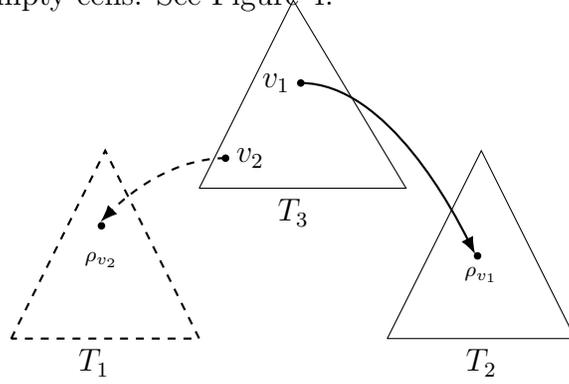
\begin{figure}[h] 
\centering
\begin{tikzpicture}
 \def\a{0.5};
\def\b{\a};
\def\lw{0.1};
 \draw [thick,dashed ] (-1,0) -- (0.25,2.5) -- (1.5,0)  -- (-1,0);
 \draw (4,0) -- (5.25,2.5) -- (6.5,0)  -- (4,0);
 \draw (1.5,2) -- (2.75,4.5) -- (4.25,2)  -- (1.5,2);
 
\fill [](1.85,2.4) circle (0.05);
\node[anchor= west] at (1.85,2.4) {$v_2$};
\fill [](2.85,3.4) circle (0.05);
\node[anchor= east] at (2.85,3.4) {$v_1$};
\fill [](5.2,1.1) circle (0.05);
\node[below] at (5.25,1.1) {\scriptsize
 $ \rho_{v_1}$};
\fill [](0.2,1.5) circle (0.05);
  \node[below] at (0.2,1.3) {\scriptsize $ \rho_{v_2}$}; 
     
\draw[thick,arrows={-Latex}] (2.85,3.4) parabola (5.17,1.12);
\draw [thick,dashed,arrows={-Latex}](1.85,2.4) parabola (0.185,1.55);
\node(start)[below] at (0.1,0.0) {$T_1$};
\node(start)[below] at (5.25,0.0) {$T_2$};

\node(start)[below] at (2.75,2) {$T_3$};

\end{tikzpicture}
\caption{Data structure $D_2$}
\label{Fig:Data structure} 
\end{figure}
    
    \vspace{-0.6cm}
    The data structure $D$ uses $ O(n)$  storage and it can be constructed in $ O(n\log n)$ time. 

\vspace{-15pt}
\subsection{{\normalsize{}Update}}\label{sec:Update}
\vspace{-5pt}
In this section we describe the \texttt{Update} algorithm. Assume that the weight of point $s_i \in S$ will be updated to  $\hat{w}$.
We begin the \texttt{Update} algorithm by updating the point $s_i$ weight in $D_1$. We find the $x$-coordinate of point $s_i$ in the tree $D_1$ and update the weight of point $s_i$. 
We store the difference between the new and the old weights of point $s_i$: $temp = w(s_i)-\hat{w}$. 
In order to find a node which represented $s_i$ we need to perform at most $O\left(\log n\right)$ operations.

We keep continuing the \texttt{Update} algorithm by updating $D_2$, i.e., we update the fields of nodes in $D_2$ that were affected by the weight change of point $s_i$. We find the \textit{cell index} $(a,b)$ of the point $s_i$ and use the $(a,b)$ to find the \textit{index} $\pi(a,b)$. In $T_3$ we find the node that is associated with $\pi(a,b)$. Denote this node as $v_{cell}$. The runtime to find $v_{cell}$ is $O\left(\log n\right)$. Without loss of generality, assume that the point $s_i$ belongs to the  cell which is covered by $P$. Then, the node $v_{cell}$ has a pointer to node  in $T_1$. Denote this node as $v_{\pi(a,b)}$. The key of node $v_{\pi(a,b)}$ is equal to the total weight of points that belong to cell with \textit{index} $\pi(a,b)$, i.e. $\rho_{\pi(a,b)}$. 
Next we have to update the key of node $v_{\pi(a,b)}$. The updated key of node $v_{\pi(a,b)}$ will be calculated by $\rho_{v_{\pi(a,b)}}=\rho_{v_{\pi(a,b)}}+temp$. We may need to rebalance the $T_1$ structure too, in $O\left(\log n\right)$.
 
Let $p' \in P$ be the square covering \minCell\ after point $s_i$ weight update. We need to check whether or not the location of $p' \in P$ should be changed to cover \maxCell\ cell instead of \minCell. If the weight of \minCell\ is greater than the \maxCell\ weight, we finish. In another case, \maxCell\ $>$ \minCell, we have to replace the square  $p'$ to cover the cell \maxCell. As the result of this, we have to update the data structure $D_2$. The update of data structure $D_2$ includes the change the value of $v_{cell}$ from $c'$ to $NULL$, update the value of node $v\in T_1 $ from \minCell\ to \maxCell, and delete from $T_2$ the node associated with the cell \maxCell\ and inserting new node associated with the cell \minCell\
to $T_2$.  The runtime is $O\left(\log n\right)$.


\vspace{-15pt}
\subsection{Insertion}
\vspace{-5pt}
In  this  section  we  describe  the \texttt{Insertion} algorithm. Let $s_{new}$ be a new point to be added to $S$.  Insertion of new point into $D$ requires the following steps: 
\begin{itemize}
    \item Insert the point $s_{new}$ as a new node in $D_1$.
    \item Update the nodes of $D_2$ that are affected by a new point $s_{new}$.
\end{itemize}

 The insertion of point  $s_{new}$ into $D_1$  takes $O\left(\log n\right)$ runtime. Next we calculate \textit{cell index} of point $s_{new}$. We find the \textit{cell index} $(a,b)$ of the point $s_{new}$ and using the $(a,b)$ to find the \textit{index} $\pi(a,b)$. We insert node with key $\pi(a,b)$ in $T_3$ (or not if $\pi(a,b)$ exist), and update $T_1$ or $T_2$. There are two options depending on whether exist a node in $T_3$ with key $\pi(a,b)$.  
In the first case a node with key $\pi(a,b)$ does not exist in $T_3$. We will insert a node with key $\pi(a,b)$ into $T_3$, and node with key $\rho_{\pi(a,b)}$ is inserted into $T_2$. Also, node with key $\pi(a,b)$ from $T_1$ has a pointer to node with key $\rho_{\pi(a,b)}$ from tree $T_2$.  The  insertion of new node in $T_1$ and $T_3$ takes  $O\left(\log n\right)$ runtime. In the second case a node with key $\pi(a,b)$ exists in $T_3$. Two options possible: the cell is covered by one of the squares from set $P$ or not. Therefore we need to update node with key $\rho_{\pi(a,b)}$ in $T_2$ or $T_3$. The update takes $O\left(\log n\right)$ runtime.



 Let $p' \in P$ be the square covering \minCell. We need to check whether or not the location of $p' \in P$ should be changed to cover \maxCell\ cell instead of \minCell. The approach of replacing the square $p'$ position is described in Section \ref{sec:Update}. In this case the  update operation takes $O\left(\log n\right)$ runtime.
 
\vspace{-15pt}
\subsection{{\normalsize{}Deletion}}
\vspace{-5pt}
In this section we describe the \texttt{Deletion} algorithm. Let $s_{del}$ be a point to be deleted from $S$. The deletion of point  $s_{del}$ from $D$ requires the following steps: 
\begin{itemize}
    \item Delete an existing point $s_{del}$ from $D_1$.
    \item Update the nodes of $D_2$ that affected by deletion of the point $s_{del}$.
\end{itemize}
The  required runtime  to delete $s_{del}$ from $D_1$ is $O\left(\log n\right)$. Next we calculate \textit{cell index} of point $s_{del}$. We find the \textit{cell index} $(a,b)$ of the point $s_{del}$ and using the $(a,b)$ to find the \textit{index} $\pi(a,b)$. We find node with key $\pi(a,b)$ in $T_3$. Next, we update $T_1$ or $T_2$, depending on the value (center of covering square or NULL) of node with key $\pi(a,b)$. Assume without loos of generality that the cell with  \textit{cell index} $(a,b)$ is covered by square $p'\in P$. We update the node with key $\rho_{\pi(a,b)} $ in $T_2$ using similar approach as described in \texttt{Update} algorithm with one difference: $temp=-w(s_{del})$. Additionally, we need to check whether or not the location of square which  covers \minCell\ should be changed to cover \maxCell\ cell instead of \minCell. 

Now, we focus on the case of disks. To deal with dynamic set $S$, when the covering shape is disk, we perform the following changes. At the start of the algorithm, we place set $P$ of $m$ disks by $HSmax_{\medcircle}(S,m)$. The disks in our solution may overlap. Therefore we have to exclude double counting of the weight of point covered by more than one disk. For this, when we calculate the weight of disk, we take into account only the points located in the cell, which is circumscribed by the disk. The rest remains the same.
 
 \begin{theorem}
Let $S$ be a set of weighted (non-negative) points on the plane, and assume that the size of $S$ never exceeds $O(n)$. It is possible to construct in time $O(n \log n)$ a data structure of size $O(n )$ that enables us to maintain a set of $m$ disks (squares), under insertions and deletions of points to/from $S$, that covers a subset of $m$ cells  having largest weights in given $G_r$, in time $O( \log n)$ per update.\qed
\end{theorem}
 
\begin{theorem}\label{THM:ApprDynami}
The approximation ratio of algorithms that solve \NAMEH\ and  \NAMED, on dynamic set of points $S$ is  $\frac{1}{4}$ or $\frac{1}{7}$, respectively.
\end{theorem}
\begin{proof}
We claim that after each change in the set of points $S$, the set of $m$ covering shapes (squares or disks) $P$  covers $m$ cells having largest weights. After any change in set $S$, this change affects only one cell. Therefore, the total weight of points that belong to this cell may be changed.  Denote this cell as \cell.
First, we assume that square/disk $p'\in SOL$ fully covers \cell\ (in the case $p'$ is square then $p'$ coincides with \cell, in the case $p'$ is disk,  \cell\ is circumscribed by $p'$). Only in the case, a total weight of points belonging to \cell\ has been decreased and became smaller than the weight of points belongs to \maxCell, we need to replace $p'$ to cover \maxCell. Opposite, we assume that no square/disk from $P$ covers \cell. Only in the case that the total weight of points that belong to \cell\ has grown up and became  greater than the weight of points that belong to \minCell, we  replace the location of one of the squares/disks from a set $P$ to coincide with \maxCell\ (or be circumscribed by $p'$ in the case of disk) instead of \minCell. After the replacement, the set $P$ covers $m$ cells with largest weight.

Thus, following the result from Theorem \ref{THM:Appr} or Theorem \ref{THM:Apprd},  we may conclude that we can achieve approximation ratio for  \NAMEH\ or \NAMED\  to be $\frac{1}{4}$ or $\frac{1}{7}$, respectively.
\end{proof}

\section{Conclusion}
\vspace{-0.3cm}
In this paper, we dealt with the problem of dynamic maintenance of drones' swarm in order to cover ground users. We presented a heuristic algorithm that provides $\frac{1}{4}$ or $\frac{1}{7}$ approximation of the optimal solution when the covered shapes are squares or disks, respectively. The time required to maintain the set of drones after any change is $O(\log n)$. It would be interesting the generalize the results for the case where UAVs are required to be connected.

\printbibliography


\end{document}